\documentclass[journal,onecolumn,12pt]{IEEEtran}
\usepackage{pslatex}
\usepackage{amsfonts,color,morefloats}
\usepackage{amssymb,amsmath,latexsym,amsthm}

\newcommand\myatop[2]{\genfrac{}{}{0pt}{}{#1}{#2}}

\newcommand{\wt}{{\mathrm{wt}}}

\newcommand{\lcm}{{\rm lcm}}

\newcommand{\ord}{{\mathrm{ord}}}
\newcommand{\Z}{\mathbb{{Z}}}

\newcommand{\cR}{{\mathcal{R}}}

\newcommand{\gf}{{\mathrm{GF}}}

\newcommand{\C}{{\mathcal{C}}}

\newcommand{\m}{{\mathrm{m}}}
\newcommand{\leader}{{\mathrm{Leader}}}

\newcommand{\bc}{{\mathbf{c}}}

\newcommand{\bb}{{\mathbf{b}}}

\newcommand{\bzero}{{\mathbf{0}}}

\newtheorem{theorem}{Theorem}
\newtheorem{lemma}[theorem]{Lemma}

\newtheorem{corollary}[theorem]{Corollary}
\newtheorem{conj}{Conjecture}
\newtheorem{problem}{Open Problem}

\newtheorem{example}{Example}

\begin{document}

\title{LCD Cyclic Codes over Finite Fields\thanks{C. Ding's research was supported by The Hong Kong Research Grants Council, Project No. 16301114.}}

\author{Chengju Li, Cunsheng~Ding, and ~Shuxing~Li    % <-this % stops a space

\thanks{C. Li is with the School of Computer Science and Software Engineering, East China Normal University,
Shanghai, 200062, China (email: lichengju1987@163.com).}% <-this % stops a space

\thanks{C. Ding is with the Department of Computer Science
                                                  and Engineering, The Hong Kong University of Science and Technology,
                                                  Clear Water Bay, Kowloon, Hong Kong, China (E-mail: cding@ust.hk).}

\thanks{S. Li is with the Department of Mathematics, The Hong Kong University of Science and Technology,
                                                  Clear Water Bay, Kowloon, Hong Kong, China  (E-mail: lsxlsxlsx1987@gmail.com).}

}

\date{\today}
\maketitle

\begin{abstract}
In addition to their applications in data storage, communications systems, and consumer electronics, LCD codes -- a class of linear codes -- have been employed in cryptography recently. LCD cyclic codes were referred to as reversible cyclic codes in the literature. The objective of this paper is to construct several families of reversible cyclic codes over finite fields and analyse their parameters. The LCD cyclic codes presented in this paper have very good parameters in general, and contain many optimal codes. A well rounded treatment of reversible cyclic codes
is also given in this paper.
\end{abstract}


\section{Introduction}\label{sec-intro}

Throughout this paper, let $q$ be a power of a prime $p$.
An $[n,k,d]$ code $\C$ over $\gf(q)$ is a $k$-dimensional subspace of $\gf(q)^n$ with minimum
(Hamming) distance $d$.
Let $\C$ be an $[n, k]$ linear code over $\gf(q)$. Its dual code, denoted by $\C^\perp$, is defined by
$$
\C^\perp =\{\bb \in \gf(q)^n: \bb \bc^T = 0 \ \forall \ \bc \in \C\},
$$
where $\bb \bc^T$ denotes the standard inner product of the two vectors $\bb$ and $\bc$.  A linear code
is called an \emph{LCD code (linear code with complementary dual)} if $\C \cap \C^\perp = \{\bzero\}$,
which is equivalent to $\C \oplus \C^\perp = \gf(q)^n$.

A linear $[n,k]$ code $\C$ over $\gf(q)$ is called {\em cyclic} if
$(c_0,c_1, \cdots, c_{n-1}) \in \C$ implies $(c_{n-1}, c_0, c_1, \cdots, c_{n-2})
\in \C$.
By identifying any vector $(c_0,c_1, \cdots, c_{n-1}) \in \gf(q)^n$
with
$$
c_0+c_1x+c_2x^2+ \cdots + c_{n-1}x^{n-1} \in \gf(q)[x]/(x^n-1),
$$
any code $\C$ of length $n$ over $\gf(q)$ corresponds to a subset of the quotient ring
$\gf(q)[x]/(x^n-1)$.
A linear code $\C$ is cyclic if and only if the corresponding subset in $\gf(q)[x]/(x^n-1)$
is an ideal of the ring $\gf(q)[x]/(x^n-1)$.

Note that every ideal of $\gf(q)[x]/(x^n-1)$ is principal. Let $\C=\langle g(x) \rangle$ be a
cyclic code, where $g(x)$ is monic and has the smallest degree among all the
generators of $\C$. Then $g(x)$ is unique and called the {\em generator polynomial,}
and $h(x)=(x^n-1)/g(x)$ is referred to as the {\em parity-check polynomial} of $\C$.

LCD cyclic codes over finite fields were called \textit{reversible codes} and studied by Massey
\cite{Massey64}. Massey showed that some LCD cyclic codes over finite fields are BCH codes, and
made a comparison between LCD codes and non-LCD codes \cite{Massey64}. He also demonstrated that
asymptotically good LCD codes exist \cite{Massey92}. Yang and Massey gave a necessary and
sufficient condition for a cyclic code to have a complementary dual \cite{YM94}. Using the hull
dimension spectra of linear codes, Sendrier showed that LCD codes meet the asymptotic Gilbert-Varshamov
bound \cite{Sendr}. Esmaeili and Yari analysed 1-generator LCD quasi-cyclic codes \cite{EY09}.
Muttoo and Lal constructed a reversible code over $\gf(q)$ \cite{ML86}.  Tzeng and Hartmann proved
that the minimum distance of a class of reversible cyclic codes is greater than the BCH bound \cite{TH70}.  Dougherty, Kim, \"Ozkaya, Sok and Sol\`e developed a linear programming bound on the largest size of
an LCD code of given length and minimum distance \cite{DKOSS}. G\"uneri, \"Ozkaya, and  Sol\`e
studied quasi-cyclic complementary dual codes \cite{GOS16}.
Carlet and Guilley investigated an
application of LCD codes against side-channel attacks, and presented several constructions of LCD
codes \cite{CG}. LCD codes can be used in a direct-sum-masking technique for the prevention of
side-channel attacks (see \cite{CG} for detail).

The objective of this paper is to construct several families of LCD cyclic codes over
finite fields and analyse their parameters. The dimensions of these codes are determined
and the minimum distances of some of the codes are settled and lower bounds on the minimum
distance of other codes are given. Many codes are optimal in the sense that they have the
best possible parameters. We will also give a well rounded treatment of LCD cyclic codes
in general.

We will compare some of the codes presented in this paper with the tables of best known linear
codes (referred to as the \emph{Database} later) maintained by
Markus Grassl at http://www.codetables.de.

\section{$q$-cyclotomic cosets modulo $n$ and auxiliaries}\label{sec-qcyclotomiccosets}

To deal with cyclic codes of length $n$ over $\gf(q)$, we have to study the canonical factorization of $x^n-1$
over $\gf(q)$. To this end, we need to introduce $q$-cyclotomic cosets modulo $n$. Note that $x^n-1$ has no
repeated factors over $\gf(q)$ if and only if $\gcd(n, q)=1$. Throughout this paper, we assume
that $\gcd(n, q)=1$.

Let $\Z_n = \{0,1,2, \cdots, n-1\}$, denoting the ring of integers modulo $n$. For any $s \in \Z_n$, the \emph{$q$-cyclotomic coset of $s$ modulo $n$\index{$q$-cyclotomic coset modulo $n$}} is defined by
$$
C_s=\{s, sq, sq^2, \cdots, sq^{\ell_s-1}\} \bmod n \subseteq \Z_n,
$$
where $\ell_s$ is the smallest positive integer such that $s \equiv s q^{\ell_s} \pmod{n}$, and is the size of the
$q$-cyclotomic coset. The smallest integer in $C_s$ is called the \emph{coset leader\index{coset leader}} of $C_s$.
Let $\Gamma_{(n,q)}$ be the set of all the coset leaders. We have then $C_s \cap C_t = \emptyset$ for any two
distinct elements $s$ and $t$ in  $\Gamma_{(n,q)}$, and
\begin{eqnarray}\label{eqn-cosetPP}
\bigcup_{s \in  \Gamma_{(n,q)} } C_s = \Z_n.
\end{eqnarray}
Hence, the distinct $q$-cyclotomic cosets modulo $n$ partition $\Z_n$.

Let $m=\ord_{n}(q)$, and let $\alpha$ be a generator of $\gf(q^m)^*$, which denotes the multiplicative group
of $\gf(q^m)$. Put $\beta=\alpha^{(q^m-1)/n}$.
Then $\beta$ is a primitive $n$-th root of unity in $\gf(q^m)$. The minimal
polynomial $\m_{s}(x)$ of $\beta^s$ over $\gf(q)$ is the monic polynomial of the smallest degree over
$\gf(q)$ with
$\beta^s$ as a zero. It is now straightforward to prove that this polynomial is given by
\begin{eqnarray}
\m_{s}(x)=\prod_{i \in C_s} (x-\beta^i) \in \gf(q)[x],
\end{eqnarray}
which is irreducible over $\gf(q)$. It then follows from (\ref{eqn-cosetPP}) that
\begin{eqnarray}\label{eqn-canonicalfact}
x^n-1=\prod_{s \in  \Gamma_{(n,q)}} \m_{s}(x)
\end{eqnarray}
which is the factorization of $x^n-1$ into irreducible factors over $\gf(q)$. This canonical factorization of $x^n-1$
over $\gf(q)$ is crucial for the study of cyclic codes.

The following result will be useful and is not hard to prove \cite[Theorem 4.1.4]{HP03}.

\begin{lemma}
The size $\ell_s$ of each $q$-cyclotomic coset $\C_s$ is a divisor of $\ord_{n}(q)$, which is the size $\ell_1$ of $C_1$.
\end{lemma}

\vspace{.1cm}
The following lemma was proved in \cite{AKS} and contains results in \cite{YH96} as special cases.

\begin{lemma}\label{lem-AKS}
Let $n$ be a positive integer such that $q^{\lfloor m/2 \rfloor}<n \leq q^m-1$, where
$m=\ord_n(q)$. Then the $q$-cyclotomic coset $C_s=\{sq^j \bmod{n}: 0 \leq j \leq m-1\}$ has cardinality
$m$ for all $s$ in the range $1 \leq s \leq n q^{\lceil m/2 \rceil}/(q^m-1)$. In addition, every $s$ with
$s \not\equiv 0 \pmod{q}$ in this range is a coset leader.
\end{lemma}

Later in this paper, we will need the following fundamental result on elementary number theory.

\begin{lemma}\label{lem-rgcd1}
Let $h \geq 1$ and let $a>1$ be an integer. Then
\begin{eqnarray*}
\gcd(a^\ell+1, a^h-1)=\left\{
\begin{array}{ll}
1 & \mbox{if $\frac{h}{\gcd(\ell, h)}$ is odd and $a$ is even,} \\
2 & \mbox{if $\frac{h}{\gcd(\ell, h)}$ is odd and $a$ is odd,} \\
a^{\gcd(\ell, h)}+1 & \mbox{if $\frac{h}{\gcd(\ell, h)}$ is even.}
\end{array}
\right.
\end{eqnarray*}
\end{lemma}

\section{Characterisations of LCD cyclic codes over finite fields}

Let $f(x)=f_{h}x^h+f_{h-1}x^{h-1}+ \cdots + f_1x+f_0$ be a polynomial over $\gf(q)$ with $f_h \ne 0$
and $f_0 \ne 0$. The reciprocal $f^*(x)$ of $f(x)$ is defined by
$$
f^*(x)=f_0^{-1}x^{h}f(x^{-1}).
$$
A polynomial is self-reciprocal if it coincides with its reciprocal. 

A code $\C$ is called {\em reversible} if $(c_0, c_1, \ldots, c_{n-1}) \in \C$ implies that
$(c_{n-1}, c_{n-2}, \ldots, c_{0})$ $\in \C$.
The conclusions of the following theorem are known in the literature (see \cite{YM94} and \cite[p. 206]{MS77}), and are easy to prove. We will employ
some of them later.

\begin{theorem}\label{thm-ReversibleCyclicCodes} {\rm\cite{YM94},\cite[p. 206]{MS77}}
Let $\C$ be a cyclic code of length $n$ over $\gf(q)$ with generator polynomial $g(x)$. Then the following statements are equivalent.
\begin{itemize}
\item $\C$ is an LCD code.
\item $g$ is self-reciprocal.
\item $\beta^{-1}$ is a root of $g$ for every root $\beta$ of $g(x)$ over the splitting field of $g(x)$.
\end{itemize}
Furthermore, if $-1$ is a power of $q$ mod $n$, then every cyclic code over $\gf(q)$ of length $n$ is reversible.
\end{theorem}

\begin{proof}
The conclusion of the last part is known \cite[p. 206]{MS77}. But we would present the following proof,
which provides hints for studying LCD cyclic codes in the next section.

Let $C_a$ denote the $q$-cyclotomic class modulo $n$ that contains $a$, where $0 \leq a \leq n-1$.
By assumption, $q^\ell \equiv -1 \pmod{n}$ for some positive integer $\ell$.
Then $a \equiv -a q^\ell \pmod{n}$. We deduce that $-a \in C_a$. Hence every irreducible factor of
$x^n-1$ is self-reciprocal. It follows that  every cyclic code over $\gf(q)$ of length $n$ is reversible.
\end{proof}

Massey showed that reversible cyclic codes are those which have self-reciprocal generator polynomials
\cite{Massey64}. It then follows from Theorem \ref{thm-ReversibleCyclicCodes} that a cyclic code is LCD if and only if it is reversible.

\section{A construction of all reversible cyclic codes over $\gf(q)$}\label{sec-allRevCodes}

The goal in this section is to give an exact count of reversible cyclic codes of length $q^m-1$ for odd primes $m$.
Recall the $q$-cyclotomic cosets $C_a$ modulo $n$ and the irreducible polynomials defined in Section
\ref{sec-qcyclotomiccosets}.
It is straightforward that $-a=n-a \in C_a$ if and only if $a(1+q^j) \equiv 0 \pmod{n}$ for some integer $j$.
The following two lemmas are straightforward.

\begin{lemma}\label{lem-selfrecip}
The irreducible polynomial $\m_a(x)$ is self-reciprocal if and only if $n-a \in C_a$.
\end{lemma}

\begin{lemma}\label{lem-nov27}
The least common multiple $\lcm(\m_a(x), \m_{n-a}(x))$ is self-reciprocal for every $a \in \Z_n$.
\end{lemma}

By Lemma \ref{lem-selfrecip}, we have that
\begin{eqnarray*}
\lcm(\m_a(x), \m_{n-a}(x))=\left\{
\begin{array}{ll}
\m_a(x) & \mbox{ if } n-a \in C_a, \\
\m_a(x)\m_{n-a}(x) & \mbox{ otherwise.}
\end{array}
\right.
\end{eqnarray*}

Let
\begin{eqnarray*}
\Pi_{(n,q)}=\Gamma_{(n,q)} \setminus \{\max\{a, \leader(n-a)\}: a \in \Gamma_{(n,q)}, n-a \not\in C_a \},
\end{eqnarray*}
where $\leader(i)$ denotes the coset leader of $C_i$.
Then $\{C_a \cup C_{n-a}: a \in \Pi_{(q,n)}\}$ is a partition of $\Z_n$.

The following conclusion follows directly from Lemmas \ref{lem-selfrecip}, \ref{lem-nov27}, and
Theorem \ref{thm-ReversibleCyclicCodes}.

\begin{theorem}\label{thm-numb}
The total number of reversible cyclic codes over $\gf(q)$ of length $n$ is equal to $2^{|\Pi_{(q,n)}|}-1$.

Every reversible cyclic code over $\gf(q)$ of length $n$ is generated by a polynomial
$$
g(x)=\prod_{a \in S} \lcm\left(\m_a(x), \m_{n-a}(x)\right),
$$
where $S$ is a nonempty subset of $\Pi_{(q,n)}$.
\end{theorem}

\begin{example}
Let $(n,q)=(15,2)$. There are the following $2$-cyclotomic classes
\begin{eqnarray*}
C_0 &=& \{0\},  \\
C_1 &=&  \{ 1, 2, 4, 8 \},\\
C_3 &=&  \{ 3, 6, 9, 12 \},\\
C_5 &=&  \{ 5, 10 \}, \\
C_7 &=&  \{ 7, 11, 13, 14 \}.
\end{eqnarray*}
We have also
$$
x^{15}-1=\m_0(x)\m_1(x)\m_3(x)\m_5(x)\m_7(x),
$$
where
\begin{eqnarray*}
\m_0(x) &=& x+1, \\
\m_1(x) &=& x^4 + x + 1, \\
\m_3(x) &=& x^4 + x^3 + x^2 + x + 1, \\
\m_5(x) &=& x^2 + x + 1, \\
\m_7(x) &=& x^4 + x^3 + 1.
\end{eqnarray*}
Note that all $\m_i(x)$ are self-reciprocal except $\m_1(x)$ and $\m_7(x)$. In this case,
$$
\Gamma_{(n,q)}=\{0, 1, 3, 5, 7\}.
$$
But
$$
\Pi_{(n,q)}=\{0, 1, 3, 5\}.
$$
Hence, there are 15 reversible binary cyclic codes of length 15.
\end{example}

\begin{corollary}
Let $q$ be an even prime power and $n=q^m-1$. If $m$ is odd, then the only self-reciprocal irreducible divisor of $x^n-1$ over $\gf(q)$ is $x-1$. If $m$ is an odd prime, then the total number of reversible cyclic codes of length $n$ over $\gf(q)$ is equal to $2^{\frac{q^m+(m-1)q}{2m}}-1$.
\end{corollary}
\begin{proof}
Since $m$ is odd and $q$ is even, it then follows from Lemma \ref{lem-rgcd1} that $\gcd(q^j+1, q^m-1)=1$
for all $j$ with $0 \leq j \leq m-1$.
Hence $a(1+q^j) \equiv 0 \pmod{n}$ if and only if $a=0$, where $a \in \Z_n$. We then deduce that
the only self-reciprocal irreducible divisor of $x^n-1$ over $\gf(q)$ is $x-1$.

Since $m$ is a prime, the length of $q$-cyclotomic cosets module $n$ is either $1$ or $m$. Since $\gcd(q-1,q^m-1)=q-1$, there are exactly $q-1$ elements in $\Z_n$, i.e, $\{i\frac{q^m-1}{q-1} \mid 0 \le i \le q-2\}$, such that the corresponding $q$-cyclotomic cosets have length $1$. Note that $0 \in \Z_n$ corresponds to $x-1$, which is the only self-reciprocal irreducible divisor of $x^n-1$ over $\gf(q)$. Thus, we have
$$
|\Pi_{(q,n)}|=\frac{q^m-1-(q-1)}{2m}+\frac{q-2}{2}+1=\frac{q^m+(m-1)q}{2m}.
$$
Hence, in this case, the total number of reversible cyclic codes of length $n$ over $\gf(q)$ is $2^{\frac{q^m+(m-1)q}{2m}}-1$.
\end{proof}

\begin{corollary}
Let $q$ be an odd prime power and $n=q^m-1$. If $m$ is odd, then the only self-reciprocal irreducible divisor of $x^n-1$ over $\gf(q)$ are $x-1$ and $x+1$. If $m$ is an odd prime, then the total number of reversible cyclic codes of length $n$ over $\gf(q)$ is equal to $2^{\frac{q^m+(m-1)q+m}{2m}}-1$.
\end{corollary}

\begin{proof}
Since $m$ is odd and $q$ is odd, by Lemma \ref{lem-rgcd1} we have $\gcd(q^j+1, q^m-1)=2$ for all $j$ with $0 \leq j \leq m-1$.
Hence $a(1+q^j) \equiv 0 \pmod{n}$ if and only if $a=0$ or $a=n/2$, where $a \in \Z_n$. We then deduce that
the only self-reciprocal irreducible divisors of $x^n-1$ over $\gf(q)$ are $x \pm 1$.

Since $m$ is a prime, the length of $q$-cyclotomic cosets module $n$ is either $1$ or $m$. Since $\gcd(q-1,q^m-1)=q-1$, there are exactly $q-1$ elements in $\Z_n$, i.e, $\{i\frac{q^m-1}{q-1} \mid 0 \le i \le q-2\}$, such that the corresponding $q$-cyclotomic cosets have length $1$. Note that $0 \in \Z_n$ and $\frac{q^m-1}{2} \in \Z_n$ correspond to $x-1$ and $x+1$, which are the only self-reciprocal irreducible divisors of $x^n-1$ over $\gf(q)$. Thus, we have
$$
|\Pi_{(q,n)}|=\frac{q^m-1-(q-1)}{2m}+\frac{q-3}{2}+2=\frac{q^m+(m-1)q+m}{2m}.
$$
Hence, in this case, the total number of reversible cyclic codes of length $n$ over $\gf(q)$ is $2^{\frac{q^m+(m-1)q+m}{2m}}-1$.
\end{proof}

\section{BCH codes}

Let $n$ be a positive integer, and let $m=\ord_n(q)$. Let $\alpha$ be a
generator of $\gf(q^m)^*$, and put $\beta=\alpha^{(q^m-1)/n}$. Then $\beta$ is
a primitive $n$-th root of unity.

For any $i$ with $0 \leq i \leq n-1$, let $\m_i(x)$ denote the minimal polynomial of $\beta^i$
over $\gf(q)$. For any $2 \leq \delta \leq n$, define
\begin{eqnarray}\label{eqn-BCHgeneratorPolyn}
g_{(q,n,\delta,b)}(x)=\lcm(\m_{b}(x), \m_{b+1}(x), \cdots, \m_{b+\delta-2}(x)),
\end{eqnarray}
where $b$ is an integer, $\lcm$ denotes the least common multiple of these minimal polynomials, and the addition
in the subscript $b+i$ of $\m_{b+i}(x)$ always means the integer addition modulo $n$.
Let $\C_{(q, n, \delta,b)}$ denote the cyclic code of length $n$ with generator
polynomial $g_{(q, n,\delta, b)}(x)$. The $\delta$ is called a \emph{designed distance} of $\C_{(q, n, \delta,b)}$. The \emph{Bose distance}, denoted $\delta_B$, of a BCH code is the largest designed distance of the code. The BCH bound says that
$$
d\geq \delta_B \geq \delta
$$
for the code $\C_{(q, n, \delta,b)}$. Thus, determining the Bose distance may improve the lower bound on
the minimum distance of $\C_{(q, n, \delta,b)}$.

When $b=1$, the set $\C_{(q, n, \delta, b)}$ is called a \emph{narrow-sense BCH code} with \emph{designed distance} $\delta$. If $n=q^m-1$,
$\C_{(q, n, \delta, b)}$ is called  a \emph{primitive BCH code}.

The following theorem was proved in \cite{AKS} and contains results in \cite{YH96} as special cases.

\begin{theorem}\label{thm-AKS}
Let $n$ be a positive integer such that $q^{\lfloor m/2 \rfloor}<n \leq q^m-1$, where
$m=\ord_n(q)$. Then the narrow-sense BCH code $\C_{(q, n, \delta, 1)}$ of length $n$ and designed distance
$\delta$, where
$2 \leq \delta \leq \min\{\lfloor n q^{\lceil m/2 \rceil}/(q^m-1) \rfloor, n\}$, has dimension
$$
k=n-m\lceil (\delta -1)(1-1/q) \rceil.
$$
\end{theorem}

Although BCH codes are not good asymptotically, they are among the best linear codes when the length
of the codes is not very large \cite[Appendix A]{Dingbk15}. So far, we have very limited knowledge of
BCH codes, as the dimension and minimum distance of BCH codes are in general open, in spite of some
recent progress \cite{Ding,DDZ15}. It is surprising that only two papers on BCH codes of length
$q^\ell +1$ have been published in the literature.

Theorem \ref{thm-AKS} gives indeed the dimension of some BCH codes, but has the following limitations:
\begin{itemize}
\item It applies only to narrow-sense BCH codes with small designed distances. Note that most BCH codes
      are not narrow-sense codes.
\item It is useful only when $n$ is close to $q^m-1$. For example, it is not useful at all when
      $n=q^{\lfloor m/2 \rfloor}+1$.
\end{itemize}

The following three theorems follow directly from Theorem \ref{thm-ReversibleCyclicCodes} and
the definition of BCH codes, and can be viewed as corollaries
of Theorem \ref{thm-ReversibleCyclicCodes}. We will make use of them directly later.

\begin{theorem}\label{thm-revBCHcodes1}
The BCH code $\C_{(q, n, \delta,b)}$ is reversible when $b=-t$ and the designed distance is $\delta=2t+2$
for any nonnegative integer $t$.
\end{theorem}

\begin{theorem}\label{thm-revBCHcodes2}
The BCH code $\C_{(q, n, \delta,b)}$ is reversible when $n$ is odd, $b=(n-t)/2$ and the designed distance is $\delta=t+2$
for any odd integer $t$ with $1 \leq t \leq n-2$.
\end{theorem}

\begin{theorem}\label{thm-revBCHcodes3}
The BCH code $\C_{(q, n, \delta,b)}$ is reversible when $n$ is even, $b=(n-2t)/2$ and the designed distance is
$\delta=2t+2$
for any integer $t$ with $0 \leq t \leq n/2$.
\end{theorem}

For all the reversible BCH codes described in Theorems \ref{thm-revBCHcodes1}, \ref{thm-revBCHcodes2},
and \ref{thm-revBCHcodes3}, we have obviously the BCH bound on the minimum distance $d \geq \delta$.
Little is known about their dimensions. Determining the dimension is a very hard problem in general.
We will settle the dimension for some of them in some special cases later.

\section{Some reversible BCH codes of length $q^{\ell}+1$ over $\gf(q)$ and their parameters}

It follows from Theorem \ref{thm-ReversibleCyclicCodes} that every cyclic code of length $n=q^\ell+1$ over $\gf(q)$
is reversible. Little has been done so far for cyclic codes of length $n=q^\ell+1$ over $\gf(q)$. Only a few papers on such codes are available in the literature. This is because the structure of the $q$-cyclotomic cosets modulo $n$ is extremely complex. However, we mention that Zetterberg's double-error correcting binary codes have length $2^\ell+1$ \cite[p. 206]{MS77}.

In this section, we will determine the dimensions of a few families of such reversible cyclic codes and improve
the BCH bound on their minimum distances by making use of the reversibility.
Throughout this section, let $m=2\ell$ and $n=q^\ell +1$.

\subsection{A basic result on $q$-cyclotomic cosets modulo $n$}

The following is basic result and will be employed very often.

\begin{lemma}
$\ord_n(q)=2\ell=m$.
\end{lemma}

\begin{proof}
Let $h$ be the least positive integer with $q^h \equiv 1 \pmod{n}$. Then $q^{\ell}+1$ divides $q^h-1$.
The desired conclusion then follows from Lemma \ref{lem-rgcd1}.
\end{proof}

The following lemma will play an important role in this section.

\begin{lemma}\label{lem-antiCosets}
Let $\ell \geq 2$. Then every positive integer $a \leq q^{\lfloor (\ell -1)/2 \rfloor}+1$
and $a \not\equiv 0 \pmod{q}$ is a coset leader and $|C_a|=2\ell$, and all the remaining positive
integers in this range are not coset leaders. In particular, these $C_a$'s
are pairwise disjoint for all such $a$'s.
\end{lemma}

\begin{proof}
We prove the conclusions of this lemma only for the case that $\ell$ is odd, and omit the proof
of the conclusions for $\ell$ being even, which is similar.

Let $\ell$ be odd from now on. Define $h=\lfloor (\ell -1)/2 \rfloor=(\ell-1)/2$. We have then $\ell=2h+1$.
Recall that $n=q^\ell+1=q^{2h+1}+1$.
We first prove that $a:=q^h+1$ is a coset leader and $|C_a|=m=2\ell$. It can be verified that
\begin{eqnarray*}
aq^j \bmod{n} = \left\{
\begin{array}{ll}
aq^j                             & \mbox{ if $1 \leq j \leq h$,} \\
(q^{h+1}-1)q^{j-(h+1)}                             & \mbox{ if $h+1 \leq j \leq 2h$,} \\
n-(q^{h}+1)q^{j-(2h+1)}                             & \mbox{ if $2h+1 \leq j \leq 3h$,} \\
n-q^{2h}-q^h                     & \mbox{ if $j=3h+1$,} \\  
n+q^{j-(3h+2)}-q^{j-(2h+1)}    & \mbox{ if $3h+2 \leq j \leq 4h+1$.}
\end{array}
\right.
\end{eqnarray*}
One can then easily check that $aq^j \bmod{n} >a$ for all $j$ with $1 \leq j \leq m-1=2\ell-1=4h+1$.
We then deduce that $a$ is a coset leader and $|C_a|=m=2\ell$.

Let $a$ be an integer with $a \not\equiv 0 \pmod{q}$ and $1 \leq a \leq q^{h}-1$. Then $a$ can
be uniquely expressed as
\begin{eqnarray}\label{eqn-q-adicexpress}
a=\sum_{u=0}^t a_u q^{i_u},
\end{eqnarray}
where
\begin{eqnarray}\label{eqn-condit}
\left\{
\begin{array}{l}
0 \leq t \leq h-1, \\
i_0=0, \\
1 \leq i_1 < i_2 < \cdots < i_t \leq h-1, \\
1 \leq a_i \leq q-1 \mbox{ for all $i$ with $0 \leq i \leq t$.}
\end{array}
\right.
\end{eqnarray}
It then follows that
\begin{eqnarray}\label{eqn-nnewnew1}
1 \leq i_{k_2} - i_{k_1} \leq h-2 \mbox{ for all $k_2 > k_1 \geq 1$}
\end{eqnarray}
and
\begin{eqnarray}\label{eqn-nnewnew2}
1 \leq i_{k} - i_{0} \leq h-1 \mbox{ for all $k \geq 1$.}
\end{eqnarray}

We now prove that $aq^j \bmod{n} > a$ for all $j$ with $1 \leq j \leq 4h+1$ by distinguishing the
following four cases.

\subsubsection*{Case I: $1 \leq j \leq h+1$}
In this case, we have clearly that $aq^j \bmod{n} = aq^j > a$.

\subsubsection*{Case II: $h+2 \leq j \leq 2h$}

If $j+i_k \leq 2h$ for all $k$ with $1 \leq k \leq t$, we have then $aq^j \bmod{n} = aq^j > a$.
Otherwise, let $k$ be the smallest such that $j+i_k \geq 2h+1$. We have then $1 \leq k \leq t$,
as $i_0+j=j \leq 2h$. By assumption, we have
$$
j+i_{u} < 2h+1 \mbox{ for $u \leq k-1$}
$$
and
$$
j+i_{u} \geq 2h+1 \mbox{ for $u \geq k$.}
$$

In this case, we have
\begin{eqnarray}\label{eqn-case2.2}
aq^j \bmod{n} = \sum_{u=0}^{k-1} a_u q^{i_u+j} -\sum_{u=k}^t a_u q^{i_u+j-2h-1}.
\end{eqnarray}
Notice that
$$
2h \geq i_{k-1}+j -(i_t +j -2h-1)=2h+1-(i_t-i_{k-1}) \geq h+2.
$$
We see that the right-hand side of \eqref{eqn-case2.2} is less than $n$ and larger than $a$.

\subsubsection*{Case III: $2h+1 \leq j \leq 3h$}

In this case, we have
$$
aq^j \equiv - \sum_{u=0}^t a_u q^{i_u+j-2h-1} \bmod{n}.
$$
Note that
$$
0 \leq i_t+j-2h-1 \leq 2h-2.
$$
We get that
\begin{eqnarray*}
aq^j \bmod{n} &=& q^{2h+1} +1 - \sum_{u=0}^t a_u q^{i_u+j-2h-1} \\
& \geq & q^{2h+1} +1 -(q-1) \sum_{u=2h-t-2}^{2h-2} q^u \\
&=& q^{2h+1} -q^{2h-1} +1 + q^{2h-t-2} \\
&>& q^h+1 \\
&>& a.
\end{eqnarray*}

\subsubsection*{Case IV: $3h+1 \leq j \leq 4h+1$}

Put $\bar{h}=j-2h-1$. Then $h \leq \bar{h} \leq 2h$. In this case, we have
$$
aq^j \equiv - \sum_{u=0}^t a_u q^{i_u+\bar{h}} \pmod{n}.
$$

If $i_u+\bar{h} \leq 2h$ for all $u$ with $1 \leq u \leq t$, then
\begin{eqnarray*}
aq^j \bmod{n} &=& q^{2h+1} +1 - \sum_{u=0}^t a_u q^{i_u+\bar{h}} \\
& \geq & q^{2h+1} +1 -(q-1) \sum_{u=2h-t}^{2h} q^u \\
&=&  1 + q^{2h-t} \\
& \geq & q^{h+1}+1 \\
&>& a.
\end{eqnarray*}

Otherwise, let $k$ be the smallest such that $i_k+\bar{h} \geq 2h+1$. We have then
$1 \leq k \leq t$, as $i_0+\bar{h}=\bar{h} \leq 2h$. Define
$$
T_1=\sum_{u=0}^{k-1} a_u q^{i_u+\bar{h}}
$$
and
$$
T_2=\sum_{u=k}^{t} a_u q^{i_u+\bar{h} -(2h+1)} \geq a_k \geq 1.
$$
We have then
$$
aq^j \equiv  -T_1 + T_2 \pmod{n}.
$$
Observe that
$$
i_t + \bar{h} -(2h+1) = i_t + j -(4h+2) \leq h-2 
$$
and 
$$ 
i_0+\bar{h}=j-(2h+1) \geq h. 
$$
We conclude that $a_0 q^{i_0 + \bar{h}} >T_2$. As a result, $-T_1+T_2 <0$. We obtain that
\begin{eqnarray*}
aq^j \bmod{n}
&=& n- T_1 + T_2 \\
&=& n - \sum_{u=0}^{k-1} a_u q^{i_u+\bar{h}} + T_2 \\
&\geq & n - (q-1) \sum_{u=2h-k+1}^{2h} q^u +T_2 \\
&=& q^{2h-k+1} +1 +T_2 \\
& \geq &  q^{2h-k+1} +1 + 1 \\
&\geq & q^{h+2}+2 \\
&>& a.
\end{eqnarray*}

Summarizing all the discussions above, we obtain the desired conclusions.
\end{proof}

\subsection{Reversible BCH codes over $\gf(q)$ of length $n=q^\ell+1$}

Since $n=q^\ell+1$ by assumption, the BCH codes $\C_{(q,m,\delta,b)}$ are reversible, and have the BCH bound $d \geq \delta$
for their minimum distances. For some of these BCH codes, we have the following bound, which is much better
when $\delta$ is getting large.

\begin{theorem}\label{thm-antiBCHbound}
Let $n=q^\ell+1$ and $m=2\ell$. Then the code $\C_{(q,n,\delta,0)}$ has minimum distance $d \geq 2(\delta-1)$.
\end{theorem}

\begin{proof}
Let $\beta=\alpha^{q^\ell-1}$, where $\alpha$ is a generator of $\gf(q^m)^*$. By definition, the generator
polynomial $g_{(q,n,\delta,0)}(x)$ of this code defined in (\ref{eqn-BCHgeneratorPolyn}) has the roots
$\beta^i$ for all $i$ in the set
$$
\{0,1,2, \cdots, \delta-2\}.
$$
It follows from Theorem \ref{thm-ReversibleCyclicCodes} that this code is reversible. As a result, the polynomial
$g_{(q,n,\delta,0)}(x)$ has the roots $\beta^i$ for all $i$ in the set
$$
\{n-(\delta-2), \cdots, n-2, n-1, 0,1,2, \cdots, \delta-2\}.
$$
Again by the BCH bound, we deduce that $d \geq 2(\delta-1)$.
\end{proof}

Given this much improved lower bound on the minimum distance of the codes $\C_{(q,n,\delta,0)}$, we would
like to determine the dimension of these codes. Unfortunately, Lemma \ref{lem-AKS} and Theorem \ref{thm-AKS}
are useless in this case
because
$$
\left\lfloor n q^{\lceil m/2 \rceil}/(q^m-1) \right\rfloor =1.
$$

The lower bound on the minimum distance of the reversible BCH code $\C_{(q,n,\delta,0)}$ is quite tight
according to experimental data. However, the determination of the dimension of this code is in general
very difficult. We will settle the dimension of this code in a number of special cases in this section.

The main result of this subsection is documented in the following theorem.

\begin{theorem}\label{thm-antidelta}
For any integer $\delta$ with $3 \leq \delta \leq q^{\lfloor (\ell-1)/2 \rfloor}+3$, the reversible code
$\C_{(q,n,\delta,0)}$ has parameters
$$
\left[q^\ell +1, \ q^\ell -2\ell \left(\delta -2 - \left\lfloor \frac{\delta-2}{q} \right\rfloor   \right), \
d \geq 2(\delta-1)   \right]
$$
and generator polynomial
$$
(x-1) \prod_{1 \leq a \leq \delta -2, \, \, a \not\equiv 0 \pmod{q}} \m_a(x),
$$
where $\m_a(x)$ is the minimal polynomial of $\beta^a$ over $\gf(q)$ and $\beta$ is the $n$-th root of unity in
$\gf(q^m)$.
\end{theorem}

\begin{proof}
Note that $0 \leq \delta-2  \leq q^{\lfloor (\ell-1)/2 \rfloor}+1$. By Lemma \ref{lem-antiCosets},
every integer $a$ with $1 \leq a \leq \delta -2$ and $a \not\equiv 0 \pmod{q}$ is a coset leader and
all the remaining integers in this range are not coset leaders. The total number of integers $a$ such
that  $1 \leq a \leq \delta -2$ and $a \equiv 0 \pmod{q}$ is equal to $\lfloor (\delta-2)/q
\rfloor$. The conclusions on the dimension and generator polynomial then follow from Lemma \ref{lem-antiCosets}
and the definition of BCH codes. The lower bound on the minimum distance comes from Theorem \ref{thm-antiBCHbound}.
\end{proof}

As a special case of Theorem \ref{thm-antidelta}, we have the following corollaries.

\begin{corollary}\label{thm-2antidelta4}
Let $q=2$. We have the following.
\begin{itemize}
\item Let $\ell \geq 3$.
The reversible code $\C_{(2,n,4,0)}$ has parameters $[2^\ell+1, 2^\ell-2\ell, 6]$ and generator
polynomial $(x-1)\m_1(x)$.
\item Let $\ell \geq 5$.
The reversible code $\C_{(2,n,6,0)}$ has parameters $[2^\ell+1, 2^\ell-4\ell, d \geq 10]$ and generator
polynomial $(x-1)\m_1(x) \m_3(x)$.
\item Let $\ell \geq 6$.
The reversible code $\C_{(2,n,8,0)}$ has parameters $[2^\ell+1, 2^\ell-6\ell, d \geq 14]$
and generator polynomial $(x-1)\m_1(x)\m_3(x)\m_5(x)$.
\item Let $\ell \geq 7$.
The reversible code $\C_{(2,n,10,0)}$ has parameters $[2^\ell+1, 2^\ell-8\ell, d \geq 18]$ 
and generator polynomial $(x-1)\m_1(x)\m_3(x)\m_5(x)\m_7(x)$.
\end{itemize}
\end{corollary}

\begin{example}
We have the following examples for the codes of Corollary \ref{thm-2antidelta4}.
\begin{itemize}
\item When $\ell \in \{3, 4,5,6\}$, $\C_{(2,n,4,0)}$ has
parameters $[9,2,6]$, $[17,8,6]$, $[32,22,6]$, and $[65,52,6]$, respectively, which are the best
possible for cyclic codes {\rm \cite[pp. 246, 247, 250, 261]{Dingbk15}}. All these codes are optimal
linear codes according to the Database.
\item When $\ell \in \{5,6,7\}$, $\C_{(2,n,6,0)}$ has
parameters $[32,12,10]$, $[65,40,10]$, and $[129,100,10]$, respectively, which are the best
possible for cyclic codes {\rm \cite[pp. 250, 261]{Dingbk15}}.
\item When $\ell \in \{6,7,8\}$, $\C_{(2,n,8,0)}$ has
parameters $[65,28,14]$, $[129,86,14]$, and $[257,208,14]$, respectively.
The first one is the best
possible for cyclic codes {\rm \cite[p. 261]{Dingbk15}}.
\end{itemize}

\end{example}

%\begin{theorem}
%When $\ell$ is sufficiently large, the binary code $\C_{(2,n,\delta,0)}$ of Theorem
%\ref{thm-antidelta} has minimum distance $d=2(\delta-1)$.
%\end{theorem}

%\begin{proof}
%The quality that $d=2(\delta-1)$ follows from the Sphere-packing bound when $\ell$ is
%sufficiently large. We omit the details of the proof.
%\end{proof}

\begin{corollary}\label{thm-3antidelta3}
Let $q=3$. We then have the following statements.
\begin{itemize}
\item Let $\ell \geq 3$.
The reversible code $\C_{(3,n,3,0)}$ has parameters $[3^\ell+1, 3^\ell-2\ell, d \geq 4]$ and generator
polynomial $(x-1)\m_1(x)$.
\item Let $\ell \geq 3$.
The reversible code $\C_{(3,n,5,0)}$ has parameters $[3^\ell+1, 3^\ell-4\ell, d \geq 8]$ and
generator polynomial $(x-1)\m_1(x) \m_2(x)$.
\item Let $\ell \geq 6$.
The reversible code $\C_{(3,n,6,0)}$ has parameters $[3^\ell+1, 3^\ell-6\ell, d \geq 10]$ and
generator polynomial $(x-1)\m_1(x)\m_2(x) \m_4(x)$.
\end{itemize}

\end{corollary}

\begin{example}
We have the following examples of the codes of Corollary \ref{thm-3antidelta3}.
\begin{itemize}
\item When $\ell \in \{3, 4\}$, $\C_{(3,n,3,0)}$ has
parameters $[28,21,4]$, $[82,73,4]$, respectively. The former has the
best possible parameters for cyclic codes {\rm \cite[p. 301]{Dingbk15}}.
\item When $\ell \in \{3,4\}$, $\C_{(3,n,5,0)}$ has
parameters $[28,15,8]$ and $[82,65,8]$, respectively. The former has
the best possible parameters for cyclic codes {\rm \cite[p. 301]{Dingbk15}}.
\item When $\ell \in \{3,4\}$, $\C_{(3,n,6,0)}$ has
parameters $[28,9,10]$ and $[82,57,10]$, respectively.
The first one is the best
possible for cyclic codes {\rm \cite[p. 301]{Dingbk15}}. The latter has the same
parameters as the best known code in the Database.
\end{itemize}

\end{example}

\begin{conj}
The following conjectures are supported by experimental data.
\begin{itemize}
\item The code $\C_{(3,n,3,0)}$ of Corollary \ref{thm-3antidelta3} has minimum distance $d=4$.
\item The code $\C_{(3,n,5,0)}$ of Corollary \ref{thm-3antidelta3} has minimum distance
$d=8$.
\end{itemize}
\end{conj}

\section{Reversible cyclic codes of length $n=q^m-1$ over $\gf(q)$}

Throughout this section, let $n=q^m-1$ for a positive integer $m$, and let $\alpha$ be a generator
of $\gf(q^m)^*$. Our task in this section is to construct reversible cyclic codes with some known
cyclic codes. Our idea is to construct reversible cyclic codes with some known families of cyclic
codes $\C$, which are not reversible. Given a cyclic code $\C$, we wish to find out conditions
under which the even-like subcode of $\C \cap \C^\perp$ or the code $\C \cap \C^\perp$ is reversible,
where the even-like subcode of $\C \cap \C^\perp$ is defined as
$$
\{c(x) \in \C \cap \C^\perp: c(1)=0\}.
$$
A known class of reversible
cyclic codes are the Melas's double-error correcting binary codes with parameters $[2^m-1, 2^m-2m,
d\geq 5]$ \cite[p. 206]{MS77}.

%\subsection{Reversible cyclic codes from punctured generalised Reed-Muller codes}

We now employ the punctured generalised Reed-Muller codes to construct reversible cyclic codes with
the construction idea above. To this end, we need to do some preparations.

For any $i$ with $0 \leq i \leq n-1$,
define $\omega_q(i)=\sum_{j=0}^{m-1} i_j$, where $i=\sum_{j=0}^{m-1} i_j q^j$ is the
$q$-adic expansion of $i$ and each $0 \leq i_j \leq q-1$.  We define
\begin{eqnarray}\label{eqn-Index2}
I_{(q,n,t)}=\{1 \leq i \leq n-1: 1 \leq \omega_q(j) \leq t \}
\end{eqnarray}
and
$$
-I_{(q,n,t)}=\{n-a: a \in I_{(q,n,t)}\},
$$
where $t \geq 1$.

\begin{lemma}\label{lem-RMc}
If $1 \leq t \leq \lceil (q-1)m/2 \rceil-1$, then $I_{(q,n,t)} \cap (-I_{(q,n,t)}) = \emptyset$.
\end{lemma}

\begin{proof}
Note that
$$n=q^m-1=(q-1)q^{m-1} + (q-1)q^{m-2} + \cdots + (q-1)q+(q-1)q^0.
$$
Hence, we have $\omega_q(i)+\omega_q(n-i)=m(q-1)$ for all $i \in \Z_n$.

By this identity, if $i \in \Z_n$ and $\omega_q(i) \leq \lceil (q-1)m/2 \rceil-1$, then
$\omega_q(n-i) > \lceil (q-1)m/2 \rceil-1$.
The desired conclusion then follows.
\end{proof}

Let $\ell=\ell_1(q-1)+\ell_0<q(m-1)$, where $\ell_0 < q-1$.
The $\ell$-th order
\emph{punctured generalized Reed-Muller code}\index{punctured generalized Reed-Muller code}
$\cR_q(\ell, m)^*$ over $\gf(q)$ is the cyclic code of length $n=q^m-1$ with generator polynomial
\begin{eqnarray}\label{eqn-generatorpolyPGRMcode}
g_R(x):= \prod_{\myatop{1 \leq j \leq n-1}{ \omega_q(j) < (q-1)m-\ell}} (x - \alpha^j),
\end{eqnarray}
where $\alpha$ is a generator of $\gf(q^m)^*$. It is easily seen that $g_R(x)$ is a polynomial over
$\gf(q)$.

By definition, we have
$$
(q-1)m-\ell = (m-\ell_1-1)(q-1)+(q-1-\ell_0).
$$
Let $h$ be the smallest integer with $\omega_q(h)=(q-1)m-\ell.$ Then
\begin{eqnarray*}
h &=& (q-1-\ell_0)q^{m-\ell_1-1} + \sum_{i=0}^{m-\ell_1-2} (q-1)q^i \\
   &=& (q-\ell_0)q^{m-\ell_1-1}-1.
\end{eqnarray*}
By the construction of the code $\cR_q(\ell, m)^*$, every integer $u$ with $0 < u <h$ satisfies $\omega_q(u)
<(q-1)m-\ell.$ Hence, the elements $\alpha^1, \alpha^2, \ldots, \alpha^{h-1}$ are all roots
of the generator polynomial $g_R(x)$ of (\ref{eqn-generatorpolyPGRMcode}). Consequently, the minimum distance
of $\cR_q(\ell, m)^*$ is at least $h$. It was proved in \cite[Theorem 5.4.1]{AssmusKey92} that the minimum
distance of $\cR_q(\ell, m)^*$ equals $h$ and the dimension of the code $\cR_q(\ell, m)^*$ is
equal to
\begin{eqnarray}\label{eqn-dimPRGM}
\sum_{i=0}^\ell \sum_{j=0}^{m} (-1)^j \binom{m}{j} \binom{i-jq+m-1}{i-jq}.
\end{eqnarray}

Let $g_R^*(x)$ denote the reciprocal of $g_R(x)$ defined above. Set
$$
g(x)=(x-1)\lcm(g_R(x), g_R^*(x)).
$$
Let $\cR_{(q, m, \ell)}$ denote the cyclic code of length $n$ over $\gf(q)$ with generator polynomial
$g(x)$.
We have then the following theorem.

\begin{theorem}\label{thm-bCode}
If $q(m-1)-2 \geq \ell \geq  1+(q-1)m-\lceil (q-1)m/2 \rceil$, then the code $\cR_{(q, m, \ell)}$ is reversible
and has minimum distance
$$d \geq 2((q-\ell_0)q^{m-\ell_1-1}-1)$$
and dimension
\begin{eqnarray}
2 \sum_{i=0}^\ell \sum_{j=0}^{m} (-1)^j \binom{m}{j} \binom{i-jq+m-1}{i-jq} -q^m.
\end{eqnarray}

\end{theorem}

\begin{proof}
When $q(m-1)-2 \geq \ell \geq  1+(q-1)m-\lceil (q-1)m/2 \rceil$,
it follows from Lemma \ref{lem-RMc} that $g_R(x)$ and $g_R^*(x)$ have no common roots. Consequently,
$g(x)=(x-1)g_R(x)g_R^*(x)$. Therefore,
$$
\deg(g(x))=2\deg(g_R(x))+1.
$$
The desired conclusion on the dimension of the code then follows from
the dimension of $\cR_q(\ell, m)^*$, which was given in (\ref{eqn-dimPRGM}). In this case, $g(x)$
has the roots $\alpha^i$ for all $i$ in the set
$$
\{n-(h-1), n-(h-2), \cdots, n-2, n-1, 0, 1, 2, \cdots, h-2, h-1\}.
$$
The desired conclusion on the minimum distance then follows from the BCH bound.
\end{proof}

The first part of Theorem \ref{thm-bCode} can be simplified into the following.

\begin{theorem}
When $q=2$ and $m-2 \geq \ell \geq  m- \lfloor (m-2)/2 \rfloor$, the code $\cR_{(q, m, \ell)}$
is a reversible cyclic code and has parameters
$$\left[2^m-1, \ 2^m-2\sum_{j=0}^{m-1-\ell} \binom{m}{j}, \ d \geq 2(2^{m-\ell}-1)\right].$$
\end{theorem}

\begin{example}
Let $m=5$ and let $\alpha$ be a generator of $\gf(2^5)^*$ with $\alpha^5 + \alpha^2  + 1=0$.
Then $\cR_{(2, 5, 3)}$ has parameters $[31, 20, 6]$, and generator polynomial
\begin{eqnarray*}
g(x)=x^{11} + x^{10} + x^9 + x^7 + x^6 + x^5 + x^4 + x^2 + x + 1.
\end{eqnarray*}
$\cR_{(2, 5, 3)}$ has the best possible parameters for cyclic codes {\rm \cite[p. 250]{Dingbk15}}.
Its dual code has parameters $[31, 11, 10]$, while the best binary cyclic
code of length $31$ and dimension $11$ has minimum distance $11$ {\rm \cite[p. 250]{Dingbk15}}.
\end{example}

\begin{example}
Let $m=6$ and let $\alpha$ be a generator of $\gf(2^6)^*$ with $\alpha^6 + \alpha^4 + \alpha^3 + \alpha + 1=0$.
Then $\cR_{(2, 6, 4)}$ has parameters $[63, 50, 6]$, and generator polynomial
$$
g(x)=x^{13} + x^9 + x^7 + x^6 + x^4 + 1.
$$
$\cR_{(2, 6, 4)}$ has the best possible parameters for cyclic codes {\rm \cite[p. 260]{Dingbk15}}.
Its dual code has parameters $[63, 13, 24]$, and is the best possible linear code {\rm \cite[p. 258]{Dingbk15}}.
\end{example}

\begin{example}
Let $m=6$ and let $\alpha$ be a generator of $\gf(2^6)^*$ with $\alpha^6 + \alpha^4 + \alpha^3 + \alpha + 1=0$.
Then $\cR_{(2, 6, 3)}$ has parameters $[63, 20, 14]$, and generator polynomial
\begin{eqnarray*}
g(x) &=& x^{43} + x^{42} + x^{40} + x^{37} + x^{36} + x^{35} + x^{34} + x^{33} + x^{29} + x^{25} + \\
     & & x^{22} + x^{21} + x^{18} + x^{14} + x^{10} + x^9 + x^8 + x^7 + x^6 + x^3 + x + 1.
\end{eqnarray*}
Its dual code has parameters $[63, 43, 6]$, which are the best possible parameters
{\rm \cite[p. 260]{Dingbk15}}.
\end{example}

Note that the punctured generalized Reed-Muller codes $\cR_q(\ell, m)^*$ are in general not BCH codes.
So are the reversible codes $\cR_{(q, m, \ell)}$. The following problem is open and interesting.

\begin{problem}
Is it true that the minimum distance $d=2((q-\ell_0)q^{m-\ell_1-1}-1)$ for the codes $\cR_{(q, m, \ell)}$
of Theorem \ref{thm-bCode}?
\end{problem}

\section{Two classes of reversible BCH cyclic codes of length $(q^m-1)/(q-1)$ over $\gf(q)$}

In this section, we construct a class of reversible cyclic codes from a family of projective BCH codes.
Throughout this section, $n=(q^m-1)/(q-1)$ and $q \geq 3$. We first do some preparations.

Let $\delta \geq 2$ be a positive integer. Define
$$
J_{(q,n,\delta)}=\cup_{1 \leq i \leq \delta-1} C_i
$$
and
$$
-J_{(q,n,\delta)}=\{n-a: a \in J_{(q,n,\delta)}\}.
$$

We will need the following conclusion.

\begin{lemma}\label{lem-BCHdisj2}
Let $\delta = q^e$, where $e = \lfloor (m-1)/2 \rfloor$.
Then $J_{(q,n,\delta)} \cap (-J_{(q,n,\delta)}) = \emptyset$.
\end{lemma}

\begin{proof}
Suppose on the contrary that $J_{(q,n,\delta)} \cap (-J_{(q,n,\delta)}) \neq \emptyset$.
Then there exist $a$, $1 \leq i \leq \delta-1$ and $1 \leq j \leq \delta-1$ such that
$$
a \in C_i\cap(-C_j),
$$
which implies that
$$
a \equiv iq^{\ell_1} \equiv -j q^{\ell_2} \pmod{n},
$$
where $0 \leq \ell_1 \leq m-1$ and $0 \leq \ell_2 \leq m-1$. Without loss of generality,
assume that $\ell_2 \geq \ell_1$. Then
$$
i+jq^{\ell_2 -\ell_1} \equiv 0 \pmod{n}.
$$
Let $\ell=\ell_2-\ell_1$. Then $0 \leq \ell \leq m-1$. We can further assume that
$\ell \leq \lceil (m-1)/2 \rceil$. Otherwise, we have
$$
iq^{m-\ell}+j \equiv 0 \pmod{n},
$$
where $m-\ell \leq \lceil (m-1)/2 \rceil$.

Since $\ell \leq \lceil (m-1)/2 \rceil$ by assumption and
$$
1 \leq i \leq \delta-1=q^{\lfloor (m-1)/2 \rfloor}-1 \mbox{ and }
1 \leq j \leq \delta-1=q^{\lfloor (m-1)/2 \rfloor}-1,
$$
one can verify that
$$
0 < i+jq^{\ell} < n,
$$
which shows that $i+jq^{\ell} \not\equiv 0 \pmod{n}$. This contradiction proves the lemma.
\end{proof}

One of the main results of this section is the following.

\begin{theorem}\label{thm-march24}
Let $\delta$ be an integer with $2 \leq \delta \leq q^{\lfloor (m-1)/2 \rfloor}$. Then the reversible BCH code
$\C_{(q,n,2\delta, 1-\delta)}$ is reversible and
has length $n=(q^m-1)/(q-1)$, dimension
$$
k=n-1-2m \left\lceil \frac{(\delta-1)(q-1)}{q} \right\rceil,
$$
and minimum distance $d \geq 2\delta$.
\end{theorem}

\begin{proof}
Let $g_u(x)$ denote the generator polynomial of the BCH code $\C_{(q, n, \delta,1)}$. It follows from
Lemma \ref{lem-AKS} that
$$
\deg(g_u(x))=m \left\lceil \frac{(\delta-1)(q-1)}{q} \right\rceil.
$$
Hence, $\C_{(q,n,2\delta, 1-\delta)}$ is reversible.
By definition, $\C_{(q,n,2\delta, 1-\delta)}$ has generator polynomial
$$
g(x)=\lcm(x-1, g_u(x), g_u^*(x)),
$$
where $g_u^*(x)$ is the reciprocal of $g_u(x)$. Notice that $2 \leq \delta \leq q^{\lfloor (m-1)/2 \rfloor}$.
By Lemma \ref{lem-BCHdisj2}, we deduce that
$$
g(x)=(x-1) g_u(x) g_u^*(x).
$$
The conclusion on the dimension of $\C_{(q,n,2\delta, 1-\delta)}$ then follows. The lower bound on the
minimum distance comes from the BCH bound.
\end{proof}

\begin{example}
The following are examples of the code of Theorem \ref{thm-march24}.
\begin{itemize}
\item When $(q, m, \delta)=(3,4,2)$, $\C_{(q,n,2\delta, 1-\delta)}$ has parameters $[40,31,4]$.
\item When $(q, m, \delta)=(3,4,3)$, $\C_{(q,n,2\delta, 1-\delta)}$ has parameters $[40,23,8]$,
      which are the best possible for cyclic codes {\rm \cite[p. 306]{Dingbk15}}.
\item When $(q, m, \delta)=(5,3,2)$, $\C_{(q,n,2\delta, 1-\delta)}$ has parameters $[31,24,5]$,
which are the best parameters for linear codes according to the Database.
\item When $(q, m, \delta)=(5,3,3)$, $\C_{(q,n,2\delta, 1-\delta)}$ has parameters $[31,18,8]$.
\item When $(q, m, \delta)=(5,3,4)$, $\C_{(q,n,2\delta, 1-\delta)}$ has parameters $[31,12,12]$.
\item When $(q, m, \delta)=(5,3,5)$, $\C_{(q,n,2\delta, 1-\delta)}$ has parameters $[31,6,19]$.
\item When $(q, m, \delta)=(4,4,3)$, $\C_{(q,n,2\delta, 1-\delta)}$ has parameters $[85,68,6]$.
\end{itemize}
\end{example}

\begin{lemma}\label{lem-Join2}
Let $m$ be a positive even integer and $\delta=q^{m/2}$. Define $\ell=(q^{m/2}-1)/(q-1)$. Then
$s \ell$ is a coset leader for each $1 \leq s \leq q-1$, $|C_{s\ell}|=m$ and $C_{s\ell}=-C_{s\ell}$. In addition,
$$
J_{(q,n,\delta)} \cap (-J_{(q,n,\delta)}) = \bigcup_{1 \leq s \leq q-1} C_{s\ell}.
$$
\end{lemma}

\begin{proof}
Let $m$ be even and $\bar{m}=m/2$. Recall that
$$
n=\frac{q^m-1}{q-1} \mbox{ and } \ell = \frac{q^{\bar{m}}-1}{q-1}
$$

We first prove that $s\ell$ is a coset leader and $|C_{s\ell}|=m$ for each $s$ with $1 \leq s \leq q-1$.
To this end, we consider $s\ell q^k \bmod{n}$ by distinguishing the following three cases.

\subsubsection*{Case I}

When $1 \leq k \leq \bar{m}-1$, it is obvious that
$$
s\ell q^k =s\frac{q^{\bar{m}+k}-q^k}{q-1}<n.
$$
As a result, $s\ell q^k  \bmod{n} =s\ell q^k  > s\ell$.

\subsubsection*{Case II}

When $k=\bar{m}$, $s\ell q^k \bmod{n}= n - s \ell >s\ell $.

\subsubsection*{Case III}

When $\bar{m}+1 \leq k \leq m-1$, we have
\begin{eqnarray*}
s\ell q^k & \equiv & s \left( \sum_{i=0}^{k-1-\bar{m}} q^i + \sum_{j=k}^{m-1} q^j  \right) \pmod{n} \\
               & \equiv & -s \sum_{i=k-\bar{m}}^{k-1} q^i \pmod{n}.
\end{eqnarray*}
It then follows that
\begin{eqnarray*}
s\ell q^k \bmod{n} &=& n - s \sum_{i=k-\bar{m}}^{k-1} q^i \\
&=& \frac{q^m-1 -s (q^{\bar{m}}-1)q^{k-\bar{m}}}{q-1} \\
& > & s\ell.
\end{eqnarray*}

Collecting the conclusions in Cases I, II and III yields the desired conclusions on $s\ell$ above.
We now proceed to prove the rest of the conclusions of this lemma.

Let $a$ and $b$ be two coset leaders in $J_{(q, n, \delta)}$ such that $C_a=-C_b$. Then there
exists a $j$ with $0 \leq j \leq m-1$ such that
\begin{eqnarray}\label{eqn-mmaineqn}
a+bq^j \equiv 0 \pmod{n}.
\end{eqnarray}

By assumption, $a \not\equiv 0 \pmod{q}$,  $b \not\equiv 0 \pmod{q}$, and
$$
1 \leq a \leq q^{\bar{m}}-1, \ 1 \leq b \leq q^{\bar{m}}-1.
$$
Let
$$
a=\sum_{i=0}^{\bar{m}-1} a_i q^i \mbox{ and } b=\sum_{i=0}^{\bar{m}-1} b_i q^i,
$$
where $0 \leq a_i \leq q-1$, $0 \leq b_i \leq q-1$, $a_0 \neq 0$ and $b_0 \neq 0$.

Below we continue our proof by considering the following three cases.

\subsubsection*{Case 1}

If $j \leq \bar{m}-1$, then $\bar{m}+j-1 \leq m-2$. Consequently,
$$
0 < a+bq^j = \sum_{i=\bar{m}}^{\bar{m}-1+j}  b_{i-j} q^i + \sum_{i=j}^{\bar{m}-1} (a_i + b_{i-j}) q^i + \sum_{i=0}^{j-1} a_i q^i < n.
$$
This means that
$$
a+bq^j \bmod{n} = a+bq^j \neq 0,
$$
which is contrary to \eqref{eqn-mmaineqn}.

\subsubsection*{Case 2}

If $j = \bar{m}$, then
\begin{eqnarray*}
bq^j &=& \sum_{i=1}^{\bar{m}} b_{\bar{m}-i} q^{m-i} \\
& \equiv & \sum_{i=2}^{\bar{m}} (b_{\bar{m}-i} - b_{\bar{m}-1})q^{m-i}  - b_{\bar{m}-1} \sum_{i=0}^{\bar{m}-1} q^i \pmod{n}.
\end{eqnarray*}
We then obtain
$$
a+bq^j \equiv T \pmod{n},
$$
where
$$
T= \sum_{i=2}^{\bar{m}} (b_{\bar{m}-i} - b_{\bar{m}-1})q^{m-i}  +  \sum_{i=0}^{\bar{m}-1}  (a_i - b_{\bar{m}-1}) q^i.
$$
Note that the highest power of $q$ in the expression of $T$ is at most $m-2$. We know that $-n < T <n$. It then
follows from \eqref{eqn-mmaineqn} that all the coefficients in the expression of $T$ are zero. This implies that
$$
a_0=a_1=\cdots=a_{\bar{m}-1}=b_0=b_1=\cdots=b_{\bar{m}-1}.
$$
Recall that $a_0 \neq 0$ and $b_0 \neq 0$. We then deduce that $a=b=s\ell$ for some $s$ with $1 \leq s \leq q-1$.
Furthermore, $C_{s\ell}=-C_{s\ell}$.

\subsubsection*{Case 3}

If $\bar{m}+1 \leq j \leq m-1$, then
\begin{eqnarray}\label{eqn-fifi}
bq^j &\equiv & \sum_{h=0}^{j-\bar{m}-1} b_{m-j+h} q^h +  \sum_{h=j}^{m-1} b_{h-j} q^h \pmod{n}  \nonumber \\
& \equiv & \sum_{k=j}^{m-2} (b_{k-j}-b_{m-j-1})q^k - b_{m-j-1} \sum_{h=j-\bar{m}}^{j-1} q^h + \sum_{h=0}^{j-\bar{m}-1} (b_{m-j+h}-b_{m-j-1})q^h \pmod{n}.
\end{eqnarray}

\subsubsection*{Case 3.1}

If $b_{k-j}-b_{m-j-1}=0$ for all $k$ with $j \leq k \leq m-2$, then
$$
b_0=b_1= \cdots= b_{m-j-1} \neq 0.
$$
It then follows from \eqref{eqn-fifi} that
$$
bq^j \bmod{n} = n - b_{m-j-1} \sum_{h=j-\bar{m}}^{j-1} q^h + \sum_{h=0}^{j-\bar{m}-1} (b_{m-j+h}-b_{m-j-1})q^h.
$$
Note that $\bar{m} \leq j-1 \leq m-2$ and $1 \leq a \leq q^{\bar{m}}-1$. We arrive at
$$
0 < a +(bq^j \bmod{n}) <n,
$$
which means that
$$
a + bq^j \not\equiv 0 \pmod{n}.
$$
This is contrary to \eqref{eqn-mmaineqn}.

\subsubsection*{Case 3.2}

If $b_{k-j}-b_{m-j-1} \neq 0$ for some $k$ with $j \leq k \leq m-2$, let $k$ be the largest such one.

\subsubsection*{Case 3.2.1}

If  $b_{k-j}-b_{m-j-1} < 0$, it follows from \eqref{eqn-fifi} that
$$
bq^j \bmod{n} = n +
 \sum_{h=j}^{k} (b_{h-j}-b_{m-j-1})q^h - b_{m-j-1} \sum_{h=j-\bar{m}}^{j-1} q^h + \sum_{h=0}^{j-\bar{m}-1} (b_{m-j+h}-b_{m-j-1})q^h.
$$
Recall that $j \leq k \leq m-2$ and $1 \leq a \leq q^{\bar{m}}-1$. We deduce that
$$
0 < a +(bq^j \bmod{n}) <n,
$$
which shows that
$$
a + bq^j \not\equiv 0 \pmod{n}.
$$
This is contrary to \eqref{eqn-mmaineqn}.

\subsubsection*{Case 3.2.2}

If  $b_{k-j}-b_{m-j-1} > 0$, it follows from \eqref{eqn-fifi} that
$$
bq^j \bmod{n} =
 \sum_{h=j}^{k} (b_{h-j}-b_{m-j-1})q^h - b_{m-j-1} \sum_{h=j-\bar{m}}^{j-1} q^h + \sum_{h=0}^{j-\bar{m}-1} (b_{m-j+h}-b_{m-j-1})q^h.
$$
Recall that $\bar{m} \leq k \leq m-2$ and $1 \leq a \leq q^{\bar{m}}-1$. We conclude that
$$
0 < a +(bq^j \bmod{n}) <n,
$$
which implies that
$$
a + bq^j \not\equiv 0 \pmod{n}.
$$
This is contrary to \eqref{eqn-mmaineqn}.

Summarizing all the conclusions in Cases 1, 2 and 3, we know that \eqref{eqn-mmaineqn} holds if and only if
$$
a=b=s\ell \mbox{ and } j=\frac{m}{2},
$$
where $1 \leq s \leq q-1$. This completes the proof of this lemma.
\end{proof}

\begin{lemma}\label{lem-march20}
Let $m \geq 4$ be even and $n=(q^m-1)/(q-1)$. Let $a$ be an integer such that
$q^{(m-2)/2} \leq a \leq q^{m/2}$ and $a \not\equiv 0 \pmod{q}$.
\begin{enumerate}
\item When $q$ is even, $a$ is a coset leader with $|C_a|=m$ except that
$$
a=i+1+i \frac{q^{m/2}-q}{q-1},
$$
where
$$
i \in \left\{ \frac{q}{2}, \frac{q}{2} +1, \cdots, \frac{q}{2} +\frac{q-4}{2} \right\}.
$$
\item When $q=3$, $a$ must be a coset leader.

When $q>3$ is odd, $a$ is a coset leader except that
$$
a=i+1+i \frac{q^{m/2}-q}{q-1},
$$
where
$$
i \in \left\{ \frac{q+1}{2}, \frac{q+1}{2} +1, \cdots, \frac{q+1}{2} +\frac{q-5}{2} \right\}.
$$

In addition, if $q$ is odd and $a$ is a coset leader, then $|C_a|=m$ except that
$$
a=\frac{q^{m/2}+1}{2}
$$
with $|C_a|=m/2$.
\end{enumerate}

\end{lemma}

\begin{proof}
Let $\bar{m}=m/2$. Let $a$ be such that $q^{(m-2)/2} \leq a \leq q^{m/2}$ and $a \not\equiv 0 \pmod{q}$.
Then the $q$-adic expansion of $a$ is of the form
$$
a=\sum_{i=0}^{\bar{m}-1} a_i q^i,
$$
where $0 \leq a_i \leq q-1$, $a_0 \neq 0$ and $a_{\bar{m}-1}  \neq 0$. Then
$$
aq^j=\sum_{i=0}^{\bar{m}-1} a_i q^{i+j}
$$
for all $j \geq 0$. To prove the desired conclusions of this lemma, we below consider $aq^j \bmod{n}$ for
$1 \leq j \leq m-1$ by distinguishing the following three cases.

\subsubsection*{Case 1: $1 \leq j \leq \bar{m}-1$}
In this case, $aq^j \bmod{n}= aq^j \bmod{n} >a$.

\subsubsection*{Case 2: $ j = \bar{m}$}

In this case, we have
\begin{eqnarray}\label{eqn-ccase2}
aq^j  \equiv \sum_{i=\bar{m}}^{m-2} (a_{i-\bar{m}} -a_{\bar{m}-1})q^i - a_{\bar{m}-1} \sum_{i=0}^{\bar{m}-1} q^i \pmod{n}.
\end{eqnarray}
We continue our discussions of Case 2 by distinguishing the following two subcases.

\subsubsection*{Case 2.1}
In this subcase, we assume that $a_i-a_{\bar{m}-1} =0$ for all $i$ with $0 \leq i \leq \bar{m}-2$. It then follows from \eqref{eqn-ccase2}
that
\begin{eqnarray*}
aq^j \bmod{n} = n -a_0 \frac{q^{\bar{m}}-1}{q-1} = \frac{q^m-1-a_0(q^{\bar{m}}-1)}{q-1 } =\frac{(q^{\bar{m}}-1)(q^{\bar{m}}+1-a_0)}{q-1} > a.
\end{eqnarray*}

\subsubsection*{Case 2.2}
In this subcase, let $k$ be the largest such that $a_k-a_{\bar{m}-1}  \neq 0$ and $0 \leq k \leq \bar{m}-2$. It then follows from \eqref{eqn-ccase2}
that
\begin{eqnarray}\label{eqn-ccase2.2}
aq^j  \equiv \sum_{i=0}^{k} (a_{i} -a_{\bar{m}-1})q^{\bar{m}+i} - a_{\bar{m}-1} \sum_{i=0}^{\bar{m}-1} q^i \pmod{n}.
\end{eqnarray}

\subsubsection*{Case 2.2.1}
If $a_k-a_{\bar{m}-1}  > 0$, it follows from \eqref{eqn-ccase2.2} that
\begin{eqnarray}\label{eqn-ccase2.2.1}
aq^j \bmod{n}  = \sum_{i=0}^{k} (a_{i} -a_{\bar{m}-1})q^{\bar{m}+i} - a_{\bar{m}-1} \sum_{i=0}^{\bar{m}-1} q^i.
\end{eqnarray}

When $k \geq 1$, we have that $\bar{m}+k \geq (m+2)/2$. It then follows from \eqref{eqn-ccase2.2.1} that
$aq^j \bmod{n}  \geq q^{\bar{m}} > a$.

When $k=0$, by assumption,
$$
a_1-a_{\bar{m}-1}=a_2-a_{\bar{m}-1}= \cdots = a_{\bar{m}-2}-a_{\bar{m}-1}=0
$$
and
$$
a_0-a_{\bar{m}-1}>0, \ a_0 \neq 0, \ a_{\bar{m}-1} \neq 0.
$$
Consequently,
$$
a_1=a_2=\cdots=a_{\bar{m}-2}=a_{\bar{m}-1}.
$$
By \eqref{eqn-ccase2.2.1}, we obtain
\begin{eqnarray}\label{eqn-sg1}
aq^j \bmod{n} &=& (a_0-a_{\bar{m}-1})q^{\bar{m}} - a_{\bar{m}-1} \sum_{i=0}^{\bar{m}-1} q^i \nonumber \\
                       &=& a_0q^{\bar{m}} - a_1 \frac{q^{\bar{m}+1}-1}{q-1}.
\end{eqnarray}
By definition and the discussions above, we get
\begin{eqnarray}\label{eqn-sg2}
a=a_0+a_1 \frac{q^{\bar{m}}-q}{q-1}.
\end{eqnarray}
Combining \eqref{eqn-sg1} and \eqref{eqn-sg2}, we arrive at
\begin{eqnarray}\label{eqn-sg3}
aq^j \bmod{n} -a = \left( a_0 - a_1 \frac{q+1}{q-1} \right)(q^{\bar{m}}-1).
\end{eqnarray}

If $q$ is even, then $\gcd(q-1, q+1)=1$. As a result,
\begin{eqnarray}\label{eqn-sg4}
a_0 - a_1 \frac{q+1}{q-1} \neq 0.
\end{eqnarray}
In this case, it can be verified that the total number of pairs $(a_0, a_1) \in \{1,2, \cdots, q-1\}^2$ such that
$a_0>a_1$ and
$$
a_0 - a_1 \frac{q+1}{q-1} < 0
$$
is equal to $(q-2)/2$, and those pairs are $(i+1, i)$, where
\begin{eqnarray}\label{eqn-sg5}
i \in \left\{ \frac{q}{2}, \frac{q}{2}+1, \cdots, \frac{q}{2} + \frac{q-4}{2} \right\}.
\end{eqnarray}
Consequently, all the $a$'s with  $q^{(m-2)/2} \leq a \leq q^{m/2}$ and $a \not\equiv 0 \pmod{q}$ are coset
leaders except that
$$
a=(i+1)+i \frac{q^{\bar{m}}-q}{q-1},
$$
where $i$ satisfies \eqref{eqn-sg5}.

If $q$ is odd, then $\gcd(q-1, q+1)=2$.  The only pair $(a_0, a_1) \in \{1,2, \cdots, q-1\}^2$ such that
$a_0>a_1$ and
\begin{eqnarray}\label{eqn-sg6}
a_0 - a_1 \frac{q+1}{q-1} = 0
\end{eqnarray}
is $((q+1)/2, (q-1)/2)$. In this case,
$$
a=\frac{q^{m/2}+1}{2}
$$
and $aq^{\bar{m}} \bmod{n}=a$. It then follows from the conclusion of Case 1 that this $a$ is a coset leader
with $|C_a|=m/2$.

It can be verified that the total number of pairs $(a_0, a_1) \in \{1,2, \cdots, q-1\}^2$ such that
$a_0>a_1$ and
\begin{eqnarray}\label{eqn-sg7}
a_0 - a_1 \frac{q+1}{q-1} < 0
\end{eqnarray}
is equal to $(q-3)/2$, and those pairs are $(i+1, i)$, where
\begin{eqnarray}\label{eqn-sg8}
i \in \left\{ \frac{q+1}{2}, \frac{q+1}{2}+1, \cdots, \frac{q+1}{2} + \frac{q-5}{2} \right\}.
\end{eqnarray}
Consequently, all the $a$'s with  $q^{(m-2)/2} \leq a \leq q^{m/2}$ and $a \not\equiv 0 \pmod{q}$ are coset
leaders except that
$$
a=(i+1)+i \frac{q^{\bar{m}}-q}{q-1},
$$
where $i$ satisfies \eqref{eqn-sg8}. This completes the discussions in Case 2.2.1.

\subsubsection*{Case 2.2.2}

If $a_k-a_{\bar{m}-1}  < 0$, it follows from \eqref{eqn-ccase2.2} that
\begin{eqnarray}\label{eqn-ccase2.2.2}
aq^j \bmod{n}  = n + \sum_{i=0}^{k} (a_{i} -a_{\bar{m}-1})q^{\bar{m}+i} - a_{\bar{m}-1} \sum_{i=0}^{\bar{m}-1} q^i.
\end{eqnarray}
Note that $\bar{m}+k \leq m-2$. It then follows from \eqref{eqn-ccase2.2.2} that
$$
aq^j \bmod{n} \geq q^{\bar{m}}>a.
$$

\subsubsection*{Case 3: $\bar{m}+1 \leq j \leq m -1$}

In this case, let $\bar{j}=j-(\bar{m}+1)$. Then $0 \leq \bar{j} \leq \bar{m}-2$. Note that $(q^m-1) \equiv 0 \pmod{n}$.
One can check that
\begin{eqnarray}\label{eqn-sg10}
aq^j \equiv T \bmod{n},
\end{eqnarray}
where
$$
T=\sum_{u=0}^{\bar{j}} (a_{\bar{m}-1-\bar{j}+u}-a_{\bar{m}-\bar{j}-2})q^u
    - a_{\bar{m}-\bar{j}-2} \sum_{u=\bar{j}+1}^{\bar{j}+\bar{m}} q^u + \sum_{u=\bar{j}+\bar{m}+1}^{m-2} (a_{u-(\bar{j}+\bar{m}+1)}- a_{\bar{m}-\bar{j}-2} )q^u.
$$

\subsubsection*{Case 3.1}

If $a_u- a_{\bar{m}-\bar{j}-2} =0$ for all $u$ with $0 \leq u \leq \bar{m}-\bar{j}-3$, then
$$
0 \neq a_0=a_1= \cdots = a_{\bar{m}-\bar{j}-2}.
$$
In this case,
$$
aq^j  \bmod{n} = n + \sum_{u=0}^{\bar{j}} (a_{\bar{m}-1-\bar{j}+u}-a_{\bar{m}-\bar{j}-2})q^u
    - a_{\bar{m}-\bar{j}-2} \sum_{u=\bar{j}+1}^{\bar{j}+\bar{m}} q^u.
$$
Note that $\bar{m}+\bar{j} \leq m-2$. We then deduce that
$$
aq^j \bmod{n} \geq q^{\bar{m}}>a.
$$

\subsubsection*{Case 3.2}

If $a_u- a_{\bar{m}-\bar{j}-2}  \neq 0$ for some $u$ with $0 \leq u \leq \bar{m}-\bar{j}-3$, let $k$ be the largest such $u$.
By definition,
$$
0 \leq k \leq \bar{m}-\bar{j}-3.
$$

\subsubsection*{Case 3.2.1}

If  $a_k- a_{\bar{m}-\bar{j}-2}  > 0$, then $T >0$. Note that $k \geq \bar{m}+\bar{j}+1 \geq \bar{m}+1$. We have
$$
aq^j \bmod{n} = T > a.
$$

\subsubsection*{Case 3.2.2}

If  $a_k- a_{\bar{m}-\bar{j}-2}  < 0$, then $T <0$. Note that $k + \bar{m}+\bar{j}+1 \leq m-2$. We have
$$
aq^j \bmod{n} = n-T > a.
$$

Collecting all the conclusions in Cases 1, 2 and 3, we complete the proof of this lemma.
\end{proof}

\begin{theorem}\label{thm-march301}
Let $m \geq 4$ be even and $2 \leq \delta \leq q^{m/2}$. Define
$$
\epsilon =\left\lfloor \frac{(\delta-2)(q-1)}{q^{m/2}-1}\right\rfloor.
$$
Then the BCH code
$\C_{(q,n,\delta, 1)}$
has length $n=(q^m-1)/(q-1)$, minimum distance $d \geq \delta$, and dimension
\begin{eqnarray*}
k=\left\{
\begin{array}{ll}
n-m \left\lceil (\delta-1)(q-1)/q \right\rceil+(2\epsilon -(q-2))\frac{m}{2}
& \mbox{ if $\epsilon \geq \lfloor (q-1)/2\rfloor$,} \\
n-m \left\lceil (\delta-1)(q-1)/q \right\rceil & \mbox{ if $\epsilon < \lfloor (q-1)/2\rfloor$.}
\end{array}
\right.
\end{eqnarray*}
\end{theorem}

\begin{proof}
Let $m$ be even. The lower bound on the minimum distance comes from the BCH bound. We prove
the conclusion on the dimension only for the case that $q$ is odd, and omit the proof of the
conclusion for
the other case, which is similar.

Let $q$ be odd. When $\epsilon \geq (q-1)/2$, it follows from Lemmas \ref{lem-AKS} and
\ref{lem-march20} that the total number of non-coset-leaders $b$ with $(q^{m/2}+1)/2 \leq b \leq \delta-1$
is equal to
$$
\epsilon-\frac{q+1}{2}+1=\epsilon-\frac{q-1}{2}.
$$
In this case, $\hat{a}:=(q^{m/2}+1)/2 \leq \delta-1$. Hence, $\hat{a}$ is a coset leader with
$|C_{\hat{a}}|=m/2$.
It follows again from Lemmas \ref{lem-AKS} and \ref{lem-march20} that the total number of coset
leaders $a$ with $1 \leq a \leq \delta-1$ is equal to
$$
\left\lceil \frac{(\delta-1)(q-1)}{q} \right\rceil - \left(\epsilon-\frac{q-1}{2}\right).
$$
For all these coset leaders $a$ we have $|C_a|=m$ except that $a=\hat{a}$.
The desired conclusion on the dimension then follows.

When $\epsilon < (q-1)/2$, we have $\delta-1 < (q^{m/2}+1)/2$, it follows from Lemmas \ref{lem-AKS} and \ref{lem-march20}, every integer $a$ with $1 \leq a \leq \delta-1$ and $a \not\equiv
0 \pmod{q}$ is a coset leader with $|C_a|=m$. The desired conclusion on the dimension then follows.
\end{proof}

\begin{corollary}
Let $m \geq 4$ be even and $\delta=q^{m/2}$.
Then the code
$\C_{(q,n,\delta, 1)}$
has length $n=(q^m-1)/(q-1)$, dimension
$$
k=n-q^{(m-2)/2}(q-1)m+(q-2)\frac{m}{2},
$$
and minimum distance $d \geq \delta+1$.
\end{corollary}

\begin{proof}
The conclusion on the dimension follows from Theorem \ref{thm-march301}. The
improvement on the lower bound of the minimum distance is due to the fact
that the Bose distance is $\delta+1$ in this case.
\end{proof}

\begin{theorem}\label{thm-march302}
Let $m \geq 4$ be even and $2 \leq \delta \leq q^{m/2}$. Define
$$
\epsilon =\left\lfloor \frac{(\delta-2)(q-1)}{q^{m/2}-1}\right\rfloor, \
\bar{\epsilon} =\left\lfloor \frac{(\delta-1)(q-1)}{q^{m/2}-1}\right\rfloor.
$$
Then the BCH code
$\C_{(q,n,2\delta, 1-\delta)}$ is reversible and
has length $n=(q^m-1)/(q-1)$, minimum distance $d \geq 2\delta$, and dimension
\begin{eqnarray*}
k=\left\{
\begin{array}{ll}
n-1-2m \left\lceil (\delta-1)(q-1)/q \right\rceil+(2\epsilon -(q-2))m + \bar{\epsilon} m
& \mbox{ if $\epsilon \geq \lfloor (q-1)/2\rfloor$,} \\
n-1-2m \left\lceil (\delta-1)(q-1)/q \right\rceil + \bar{\epsilon} m & \mbox{ if $\epsilon < \lfloor (q-1)/2\rfloor$.}
\end{array}
\right.
\end{eqnarray*}
\end{theorem}

\begin{proof}
Notice that the code $\C_{(q,n,2\delta, 1-\delta)}$ is reversible. The lower bound on the minimum
distance comes from the BCH bound.
Let $g_{(q, n, \delta, 1)}(x)$ denote the generator polynomial of the code $\C_{(q,n,\delta, 1)}$
of Theorem \ref{thm-march301}. It then follows from Theorem \ref{thm-march301} that
\begin{eqnarray*}
\deg(g_{(q, n, \delta, 1)}(x)) =\left\{
\begin{array}{ll}
m \left\lceil (\delta-1)(q-1)/q \right\rceil-(2\epsilon -(q-2))\frac{m}{2}
& \mbox{ if $\epsilon \geq \lfloor (q-1)/2\rfloor$,} \\
m \left\lceil (\delta-1)(q-1)/q \right\rceil & \mbox{ if $\epsilon < \lfloor (q-1)/2\rfloor$.}
\end{array}
\right.
\end{eqnarray*}

By definition, the generator polynomial $g_{(q, n, 2\delta, 1-\delta)}(x)$ of $\C_{(q,n,2\delta, 1-\delta)}$ is given by
\begin{eqnarray*}
g_{(q, n, 2\delta, 1-\delta)}(x)
&=& \lcm(x-1, g_{(q, n, \delta, 1)}(x), g_{(q, n, \delta, 1)}^*(x)) \\
&=& (x-1) \frac{g_{(q, n, \delta, 1)}(x) g_{(q, n, \delta, 1)}^*(x)}{\gcd(g_{(q, n, \delta, 1)}(x), g_{(q, n, \delta, 1)}^*(x))},
\end{eqnarray*}
where $g_{(q, n, \delta, 1)}^*(x)$ is the reciprocal of $g_{(q, n, \delta, 1)}(x)$.
Consequently,
$$
\deg(g_{(q, n, 2\delta, 1-\delta)}(x))=1+2 \deg(g_{(q, n, \delta, 1)}(x))-\deg(\gcd(g_{(q, n, \delta, 1)}(x), g_{(q, n, \delta, 1)}^*(x))).
$$
By Lemma \ref{lem-Join2}, we have
$$
\deg(\gcd(g_{(q, n, \delta, 1)}(x), g_{(q, n, \delta, 1)}^*(x)))=\bar{\epsilon}.
$$
The desired conclusion on the dimension of $\C_{(q,n,2\delta, 1-\delta)}$ then follows.
\end{proof}

For the two parameters $\bar{\epsilon}$ and $\epsilon$ defined in Theorem \ref{thm-march302},
we have $\bar{\epsilon}=\epsilon$ except in a few cases where $\bar{\epsilon}=\epsilon+1$.

\begin{corollary}\label{cor-lastone}
Let $m \geq 4$ be even and $\delta=q^{m/2}$.
Then the reversible BCH code
$\C_{(q,n,2\delta, 1-\delta)}$
has length $n=(q^m-1)/(q-1)$, dimension
$$
k=n-1-2mq^{(m-2)/2}(q-1)+(2q-3)m,
$$
and minimum distance $d \geq 2\delta+2$.
\end{corollary}

\begin{proof}
The conclusion on the dimension follows from Theorem \ref{thm-march302}. The
improvement on the lower bound of the minimum distance is due to the fact
that the Bose distance is $\delta+1$ in this case.
\end{proof}

\begin{example}
Let $(q, m, \delta)=(3, 4, 9)$. Then the code $\C_{(q,n,2\delta, 1-\delta)}$ has parameters $[40,3, 20]$.
\end{example}

\begin{example}
Let $(q, m, \delta)=(4, 4, 16)$. Then the code $\C_{(q,n,2\delta, 1-\delta)}$ has parameters $[85, 8, 34]$.
\end{example}

\section{Concluding remarks}

The main contributions of this paper are the following:
\begin{itemize}
\item The construction of all reversible cyclic codes over finite fields documented in Section \ref{sec-allRevCodes}.
\item The construction of the family of reversible cyclic codes of length $n=q^\ell +1$ over $\gf(q)$ and the analysis of
         their parameters (see Theorem \ref{thm-antidelta}).
\item The analysis of the family of reversible cyclic codes of length $n=q^m-1$ over $\gf(q)$ (see
         Theorem \ref{thm-bCode}).
\item The analysis of the family of reversible cyclic codes of length $n=(q^m-1)/(q-1)$ over $\gf(q)$
          (see Theorem \ref{thm-march302}).
\end{itemize}

The dimensions of all these codes were settled.  Lower bounds on all the reversible cyclic codes were
derived from the BCH bound. In most cases, we conjecture that the lower bounds are actually the minimum distances
of the codes. However, it is extremely difficult to determine the minimum distance of these cyclic codes. The reader
is cordially invited to settle the open problems and conjectures proposed in this paper.

\section*{Acknowledgements}

The authors are very grateful to the reviewers and the Associate Editor, Prof. Chaoping Xing, for their detailed comments and suggestions that much improved the presentation and quality of this paper.

\end{document}